\documentclass[11pt,a4paper,leqno]{amsart}

\usepackage[utf8]{inputenc}
\usepackage[english]{babel}
\usepackage{amsmath,amsfonts,amssymb,amsthm}
\usepackage{xcolor}
\usepackage{tabularx}
\usepackage{tikz}

\theoremstyle{plain}
\newtheorem{theoreme}{Theorem}[section]

\newtheorem{lemme}[theoreme]{Lemma}

\theoremstyle{definition}
\newtheorem{definition}[theoreme]{Definition}

\newtheorem{rmq}[theoreme]{Remark}

\makeatletter

\@addtoreset{equation}{section}
\makeatother

\makeatletter
\renewcommand{\thefigure}{\thesection.\@arabic\c@figure}
\@addtoreset{figure}{section}
\makeatother

\begin{document}

\author{Fabien Heuwelyckx}
\dedicatory{Université de Mons}
\address{Institut Complexys, Département de Mathématique,
Université de Mons, 20~Place du Parc, 7000 Mons, Belgium. This work
was supported by the National Bank of Belgium.}
\email{fabien.heuwelyckx@umons.ac.be}
\title[Convergence of lookback options in the binomial model]{CONVERGENCE OF 
EUROPEAN LOOKBACK OPTIONS 
WITH FLOATING STRIKE IN THE BINOMIAL MODEL}
\keywords{Binomial model, lookback, floating strike, Black-Scholes,
convergence, asymptotic}

\begin{abstract}
  In this article we study the convergence of a European lookback
  option with floating strike evaluated with the binomial model of
  Cox-Ross-Rubinstein to its evaluation with the Black-Scholes model.
  We do the same for its delta.  We confirm that these convergences
  are of order $1/\sqrt{n}$.  For this, we use the binomial model of
  Cheuk-Vorst which allows us to write the price of the option using a
  double sum. Based on an improvement of a lemma of Lin-Palmer, we are
  able to give the precise value of the term in $1/\sqrt{n}$ in the
  expansion of the error; we also obtain the value of the term
  in~$1/n$ if the risk free interest rate is non zero.  This
  modelisation will also allow us to determine the first term in the
  expansion of the delta.
\end{abstract}

\maketitle

\section{Introduction}
The goal of this paper is to study the rate of convergence for the
price and the delta of the European lookback option with floating
strike given by the Cox-Ross-Rubinstein~\cite{CRR} binomial model.
Our main results solve a problem posed by Lin and
Pal\-mer~\cite{lin_palmer}. As far as we know the convergence to the
Black-Scholes price as the number of periods tends to infinity has not
yet been proved for lookback options. This result seems plausible in
view of what is happening with other options. The existence of a
limit, however, is known from~Jiang and~Dai~\cite{jiang_dai}. We will
look at this problem in the case where the evaluation is performed at
time~$t=0$.

A European lookback call gives the holder the right to buy the
underlying at maturity for its lowest price during its lifetime. The
payoff function for the lookback call is given by
\[
C^{fl}=\max (S_T-\min_{t\leq T}S_t,0)=S_T-\min_{t\leq T}S_t,
\]
where $S_t$ is the price of underlying asset at a fixed time~$t$. A
European lookback put gives the holder the right to sell the
underlying at maturity for its highest price during its lifetime. The
payoff function for the lookback put is given by
\[
P^{fl}=\max (\max_{t\leq T}S_t-S_T,0)=\max_{t\leq T}S_t-S_T.
\]
In fact, both options will always be exercised by its holder, and they
are naturally more expensive than a standard option.

The formula for the Black-Scholes price has been obtained by Goldman,
Sosin and Gatto~\cite{goldman_sosin}. The value of the call at
time~$t=0$ is
\begin{equation}\label{formule-BS-call}
C^{fl}_{BS}=S_0\Big(1+\frac{\sigma^2}{2r}\Big)\Phi(a_1)
-S_0 e^{-rT}\Big(1-\frac{\sigma^2}{2r}\Big)\Phi(a_2)
-S_0\frac{\sigma^2}{2r},
\end{equation}
and for the put it is
\begin{equation}
P^{fl}_{BS}=C^{fl}_{BS}-S_0\Big(1-e^{-rT}\Big)\Big(1-\frac{\sigma^2}{2r}\Big),
\end{equation}
where $a_1=(r/\sigma+\sigma/2)\sqrt{T}$ and
$a_2=(r/\sigma-\sigma/2)\sqrt{T}$.  We have used the usual notation
for the different parameters: $S_0$~as the initial value of the stock
price, $r$~as the spot rate, $\sigma$~as the volatility of the
underlying asset, $T$~as the time to maturity and $\Phi$~as the
standard normal cumulative distribution function. 

The main theorems and their proofs use several times the following
notion of asymptotic expansion that admits variable but bounded
coefficients; see Diener-Diener~\cite{diener_diener}.
\begin{definition}
  Let $(f_k)_{k\geq 0}$ be a sequence of bounded functions of $n$. We
  will say that a function $f$ has an \emph{asymptotic expansion in
    powers of~$n^{-1/2}$ to the order~$m$} ($m$ an integer greater
  than~$-2$) if
\[
\sqrt{n}^{\,m+1}
\Big(f-\sum_{k=-2}^m \frac{f_i}{\sqrt{n}^{\,k}}\Big)
\]
is bounded.
\end{definition}
\noindent We will confirm that the rate of convergence of European
lookback options with floating strike is of order~$n^{-1/2}$, and we
derive formulas for the coefficients of~$n^{-1/2}$ and~$n^{-1}$ (the
latter provided that $r\neq 0$) in the asymptotic expansion.
Specifically, in section~\ref{section-price}, we give an asymptotic
expansion of the type
\begin{equation}\label{da-options}
\Pi^{fl}_n=\Pi_{BS}^{fl}+\frac{\Pi_1}{\sqrt{n}}+\frac{\Pi_2}{n}
+O\Big(\frac{1}{n^{3/2}}\Big),
\end{equation}
where $\Pi_{BS}^{fl}$ is the Black-Scholes price. In
section~\ref{section-delta} we obtain the corresponding formula for
the delta.
\section{The Cheuk-Vorst Model}
We first recall how to price a European lookback option with floating
strike in a Cox-Ross-Rubinstein~(CRR) binomial model and in particular
with the Cheuk-Vorst~\cite{cheuk_vorst} lattice with one state
variable which is equivalent to the one studied. We choose this
approach following the suggestion of Lin and Palmer~\cite{lin_palmer}.
However, there exists another approach of similar complexity given by
Föllmer and Schied~\cite{follmer_schied}, see
Appendix~\ref{annexe-equivalence}.

We use the notation~$n$ for the number of periods. Conventional
assumptions are that the proportional upward jump~$u_n$ and the
downward jump~$d_n$ are given by
\[
u_n=e^{\sigma\sqrt{T/n}}\quad\text{and}\quad
d_n=u_n^{-1}=e^{-\sigma\sqrt{T/n}}.
\]
The options of our interest are path-dependent, which means that the
traditional tree does not work because the price depends not only on
the node but also on the trajectory. Hull and White~\cite{hull_white}
have developed a binomial model with nodes subdivided into different
states. The problem is that this tree has a lot of information at
each node and needs more calculations. To remedy this, Cheuk and
Vorst have created a modified tree, still constructed by backward
induction, where the value associated with a specific node depends
only on the time and on the difference in powers of~$u_n$ between the
present and the lowest value (or highest for the put) of the
underlying from time~$t=0$ to the present time.  For the call this
difference is the value of the integer~$j$ such that
\[
S_t=\Big(\min_{t\leq T}S_t\Big) u_n^j,
\]
with $S_t$~the price of the underlying at the time~$t$. For the put
this difference is the value of the integer~$j$ such that
\[
S_t=\Big(\max_{t\leq T}S_t\Big) u_n^{-j}.
\]
For a fixed number of steps~$n$ we call $V^*_{m,n}$, with $m$ an
integer between 0 and~$n$, the random variable describing the number
of levels above the minimum after $m$~periods.  We create a similar
random variable~$W^*_{m,n}$ expressing the number of levels below the
maximum.  With all this we can write exact formulas for the prices of
the two types of European lookback options with floating strike.
\begin{theoreme}\label{CV-formulas}
  The price of a European lookback call with floating strike is
  $C^{fl}_n=S_0 V_n(0,0)$, and the price of a European lookback put
  with floating strike is $P^{fl}_n=S_0 W_n(0,0)$, with
\begin{equation}\label{def-V}
V_n(0,0)=\sum_{j=0}^{n}(1-u_n^{-j})\, P(V^*_{0,n}=0,V^*_{n,n}=j)
\end{equation}
and
\begin{equation}\label{def-W}
W_n(0,0)=\sum_{j=0}^{n}(u_n^j-1)\, P(W^*_{0,n}=0,W^*_{n,n}=j).
\end{equation}
The probabilities are given by
\begin{equation}\label{densite-V}
P(V^*_{0,n}=0, V^*_{m,n}=j)=\sum_{k=j}^l \Lambda_{j,k,m}\,q_n^k (1-q_n)^{m-k},
\end{equation}
and
\begin{equation}\label{densite-W}
P(W^*_{0,n}=0,W^*_{m,n}=j)=\sum_{k=j}^l \Lambda_{j,k,m}\,(1-q_n)^k q_n^{m-k},
\end{equation}
for $j=0,1,\dots,m$, where 
\[
l=\Big\lfloor\frac{m+j}{2}\Big\rfloor,\quad
q_n=p_n u_n e^{-r\frac{T}{n}},\quad 
p_n=\frac{e^{r\frac{T}{n}}-d_n}{u_n-d_n},
\]
and 
\[
\Lambda_{j,k,m}=\binom{m}{k-j}-\binom{m}{k-j-1}
\]
if $k>j$ and $\Lambda_{j,k,m}=1$ if $k=j$.
\end{theoreme}

\noindent Note that we use the same notation, $V_n(j,m)$, as Cheuk and Vorst for
the value associated with a node, where $j$~refers
to the level of the node and $m$~refers to the time~$mT/n$, see
Figure~\ref{fig-cheuk}.

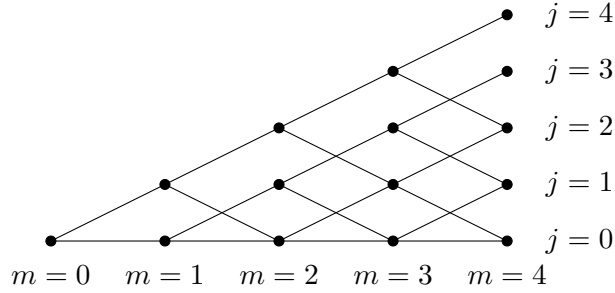
\begin{figure}[h!]
\begin{center}
\begin{tikzpicture}[scale=0.75]
\draw (0,0) -- (2,0) -- (4,0) -- (6,0) -- (8,0);
\draw (0,0) -- (2,1) -- (4,2) -- (6,3) -- (8,4);
\draw (2,0) -- (4,1) -- (6,2) -- (8,3);
\draw (2,1) -- (4,0) -- (6,1) -- (8,2);
\draw (4,1) -- (6,0) -- (8,1);
\draw (4,2) -- (6,1) -- (8,0);
\draw (6,2) -- (8,1);
\draw (6,3) -- (8,2);
\filldraw (0,0) circle(2.5pt);
\filldraw (2,0) circle(2.5pt);
\filldraw (2,1) circle(2.5pt);
\filldraw (4,0) circle(2.5pt);
\filldraw (4,1) circle(2.5pt);
\filldraw (4,2) circle(2.5pt);
\filldraw (6,0) circle(2.5pt);
\filldraw (6,1) circle(2.5pt);
\filldraw (6,2) circle(2.5pt);
\filldraw (6,3) circle(2.5pt);
\filldraw (8,0) circle(2.5pt);
\filldraw (8,1) circle(2.5pt);
\filldraw (8,2) circle(2.5pt);
\filldraw (8,3) circle(2.5pt);
\filldraw (8,4) circle(2.5pt);
\node[right] at (8.5,0){$j=0$};
\node[right] at (8.5,1){$j=1$};
\node[right] at (8.5,2){$j=2$};
\node[right] at (8.5,3){$j=3$};
\node[right] at (8.5,4){$j=4$};
\node[below] at (0,-0.25){$m=0$};
\node[below] at (2,-0.25){$m=1$};
\node[below] at (4,-0.25){$m=2$};
\node[below] at (6,-0.25){$m=3$};
\node[below] at (8,-0.25){$m=4$};
\end{tikzpicture}
\caption{Cheuk-Vorst lattice for the call with $n=4$}\label{fig-cheuk}
\end{center}
\end{figure}

We will provide the proof of this result in
Appendix~\ref{annexe-CVdensity}. In formula~(\ref{densite-V}),
respectively~(\ref{densite-W}), the variable~$k$ expresses the
different possibilities for the number of ups, respectively downs,
made by the underlying to arrive at the level~$j$ after $m$~periods.
While equation~(\ref{def-V}) can be deduced from the points $\langle
8\rangle$, $\langle 11\rangle$ and $\langle 13\rangle$
in~\cite{cheuk_vorst}, the explicit formulas~(\ref{densite-V})
and~(\ref{densite-W}) were not given by Cheuk and Vorst in their
paper. They are crucial for our analysis.

\section{Intermediate lemma}
In this section, we develop a lemma that will be very useful in the
proof of the main theorems. It will give us the asymptotic expansion
of expressions which can be written as complementary cumulative
distribution functions of binomial distributions. For a possible
future application we will obtain a more general result than is needed
for the present paper.
\begin{lemme}\label{lemme-nouveau}
Suppose that
\[
p_n=\frac{1}{2}+\frac{\alpha}{\sqrt{n}}+\frac{\beta}{n}
+\frac{\gamma}{n^{3/2}}+\frac{\delta}{n^2}
+O\Bigl(\frac{1}{n^{5/2}}\Bigr)
\]
and
\[
j_n=\frac{n}{2}+a\sqrt{n}+\frac{1}{2}+b_n+\frac{c}{\sqrt{n}}+\frac{d}{n}
+O\Bigl(\frac{1}{n^{3/2}}\Bigr),
\]
where the sequence~$(b_n)_n$ is bounded. Then
\begin{align*}
&\sum_{k=j_n}^{n}\binom{n}{k}p_n^k(1-p_n)^{n-k}\\
&\quad=\Phi(A)+\frac{e^{-A^2/2}}{\sqrt{2\pi}}
\Bigl(\frac{B_n}{\sqrt{n}}+\frac{C_0\!-\!C_2 B_n^2}{n}
+\frac{D_0\!-\!D_1 B_n\!-\!D_3 B_n^3}{n^{3/2}}\Bigr)
+O\Bigl(\frac{1}{n^2}\Bigr),
\end{align*}
where 
\begin{itemize}
\item[$A=$\!] $2(\alpha-a)$,
\item[$B_n=$\!] $2(\beta-b_n)$,
\item[$C_0=$\!] $2(\alpha^2 A+\gamma-c)+(2\alpha/3-A/12)(1-A^2)$,
\item[$C_2=$\!] $A/2$,
\item[$D_0=$\!] $2(2\alpha\beta A+\delta-d)+2(1-A^2)\beta/3$,
\item[$D_1=$\!] $(1-4A^2+A^4)/12-2\alpha(\alpha-A-\alpha A^2+A^3/3)
+2A(\gamma-c)$,
\item[$D_3=$\!] $(1-A^2)/6$,
\end{itemize}
and $\Phi$ is the cumulative distribution function of the standard normal
distribution.
\end{lemme}

The proof of this lemma (see Appendix~\ref{annexe-proof-lemma-new})
will be based on the following result, which is a simplified version
of a lemma of Lin and Palmer~\cite{lin_palmer}; its proof will be
given in Appendix~\ref{annexe-proof-lemma}.

\begin{lemme}\label{lemme-intermediaire}
Provided $p_n=1/2\,+\,O(n^{-1/2})$ 
as $n\to\infty$ and $0\leq j_n\leq n+1$ for $n$ sufficiently large, then
\begin{equation}\label{formule-lemma}
\begin{split}
&\sum_{k=j_n}^{n}\binom{n}{k}p_n^k(1-p_n)^{n-k}
= \frac{1}{\sqrt{2\pi}}\int_{\xi}^{\infty}\!\!e^{-u^2/2}\,\mathrm{d}u\\
&\quad -\frac{(1-\xi^2)e^{-\xi^2/2}}{6\sqrt{2\pi}}
\frac{1-2p_n}{\sqrt{p_n(1-p_n)}}
\frac{1}{\sqrt{n}}
+\frac{\xi(1-\xi^2)e^{-\xi^2/2}}{12\sqrt{2\pi}}\frac{1}{n}
+O\Bigl(\frac{1}{n^2}\Bigr),
\end{split}
\end{equation}
where $\xi=\frac{j_n-np_n-1/2}{\sqrt{np_n(1-p_n)}}$.
\end{lemme}

As previously announced we will use a special situation of this lemma
to prove the main results, since the values of $\beta$, $\delta$, $a$,
$c$ and~$d$ will be zero in the case that we are interested in. Note
that with this restriction we obtain $A=2\alpha$, $B_n=-2 b_n$,
$C_0=\alpha /2+2\alpha^3+2\gamma$, $C_2=\alpha$, $D_0=0$,
$D_1=(1+8\alpha^2+48\alpha^4)/12 +4\alpha\gamma$ and
$D_3=(1-4\alpha^2)/6$.
\section{Convergence of the price}\label{section-price}
In this section, we give the first main results of this paper, where
we obtain the convergence of the CRR binomial model to the
Black-Scholes formula and the first two coefficients in the asymptotic
expansion in powers of~$n^{-1/2}$ of the lookback option if the risk
free interest rate is not null.
These coefficients are constants. 
With this finding, the standard Richardson extrapolation could be used to obtain
solutions with a~higher convergence rate.

In contrast, the coefficients are not necessarily constant in the case of
other options as Diener-Diener~\cite{diener_diener} and
Walsh~\cite{walsh} demonstrated for vanilla options.
For these options, it is well known that there are oscillations.
For barrier options, Gobet~\cite{gobet} has proved zigzag convergence, while Lin and
Palmer~\cite{lin_palmer} have obtained an asymptotic expansion with non constant 
coefficients.

If we compute the price of a lookback option with floating strike at emission~$t=0$, 
there is no zigzag convergence; however, as numerical experiments show, for a 
general emission time~$t$ there will be oscillations.

\begin{theoreme}\label{thm-call}
  In the n-period CRR binomial model, let~$S_0$ be the initial stock
  price, $r$~the risk free interest rate, $\sigma$~the volatility and
  $T$~the time to maturity. If $r\neq 0$ then the asymptotic formula
  for the price of the European lookback call option with floating
  strike is
\[
\begin{split}
C_{n}^{fl}&=C^{fl}_{BS}+\frac{\sigma\sqrt{T}}{2}
(C^{fl}_{BS}-S_0)\frac{1}{\sqrt{n}}\\
&\quad +\Big[\frac{\sigma^2 T}{12}\Big(C_{BS}^{fl}
+2S_0 \Big[\Phi(a_1)-e^{-rt}\Phi(a_2)-\frac{3}{2}\,\Big]\Big)
+S_0\frac{\sigma\sqrt{T}}{2}\frac{e^{-a_1^2/2}}{\sqrt{2\pi}}\Big]\frac{1}{n}\\
&\quad +O\Bigl(\frac{1}{n^{3/2}}\Bigr),
\end{split}
\]
where $C^{fl}_{BS}$ is the price given by the Black-Scholes model,
$\Phi$ the standard normal cumulative distribution function,
$a_1=(r/\sigma+\sigma/2)\sqrt{T}$ and
$a_2=(r/\sigma-\sigma/2)\sqrt{T}$.
\end{theoreme}

\begin{proof}
  By Theorem~\ref{CV-formulas}, we need to approximate the value
\begin{equation}\label{v00}
V_n(0,0)=\sum_{j=0}^{n}(1-u_n^{-j})\sum_{k=j}^{l_1} 
\Lambda_{j,k,n}\,q_n^k (1-q_n)^{n-k}
\end{equation}
deduced from~(\ref{def-V}) and~(\ref{densite-V}) where $l_1=\lfloor
(n+j)/2\rfloor$ and $\Lambda_{j,k,n}=\binom{n}{k-j}-\binom{n}{k-j-1}$
if $k>j$ and 1 if $k=j$. First, we rewrite the inner sum in two sums
starting at 0 and get
\begin{align*}
&\sum_{k=j}^{l_1} \Lambda_{j,k,n}\,q_n^k (1-q_n)^{n-k}\\
&=\sum_{k=j}^{l_1} \binom{n}{k-j} q_n^k (1-q_n)^{n-k}
-\sum_{k=j+1}^{l_1} \binom{n}{k-j-1} q_n^k (1-q_n)^{n-k}\\
&=Q_n^j\sum_{m=0}^{l_2} \binom{n}{m} q_n^m (1-q_n)^{n-m}
-Q_n^{j+1}\sum_{m=0}^{l_2-1} \binom{n}{m} q_n^m (1-q_n)^{n-m},
\end{align*}
where $Q_n=q_n/(1-q_n)$ and $l_2=\lfloor (n-j)/2\rfloor$.
Thus~$V_n(0,0)=\varphi_1 - \varphi_2$ with
\begin{align*}
\varphi_1 &= \sum_{j=0}^{n}\sum_{m=0}^{l_2}
Q_n^j(1-u_n^{-j})\binom{n}{m} q_n^m (1-q_n)^{n-m},\\
\varphi_2 &= Q_n\sum_{j=0}^{n}\sum_{m=0}^{l_2-1}
Q_n^j(1-u_n^{-j})\binom{n}{m} q_n^m (1-q_n)^{n-m}.
\end{align*}
We interchange the sums, noting that $0\leq j\leq n$ and $0\leq m\leq
\lfloor(n-j)/2\rfloor$ is equivalent to $0\leq m\leq \lfloor
n/2\rfloor$ and $0\leq j\leq n-2m$, and similarly
for~$\varphi_2$. Thus
\begin{align*}
\varphi_1 &= \sum_{m=0}^{\lfloor n/2\rfloor}
\binom{n}{m} q_n^m (1-q_n)^{n-m}
\sum_{j=0}^{n-2m}Q_n^j(1-u_n^{-j}),\\
\varphi_2 &= Q_n\sum_{m=0}^{\lfloor n/2\rfloor-1}
\binom{n}{m} q_n^m (1-q_n)^{n-m}
\sum_{j=0}^{n-2-2m}Q_n^j(1-u_n^{-j}).
\end{align*}
Using the formula for geometric series and the link between $u_n$ and
$d_n$, we may simplify the inner sums as
\begin{equation}\label{decomposition_sommes}
\begin{split}
&\sum_{j=0}^{n-2m}Q_n^j-\sum_{j=0}^{n-2m}(Q_n d_n)^j\\
&=\frac{Q_n}{\Theta_n}((1-d_n)+(Q_n d_n-1)Q_n^{n-2m}
+d_n(1-Q_n)(Q_n d_n)^{n-2m})
\end{split}
\end{equation}
and
\begin{align*}
&\sum_{j=0}^{n-2-2m}Q_n^j-\sum_{j=0}^{n-2-2m}(Q_n d_n)^j\\
&=\frac{Q_n}{\Theta_n}((1-d_n)+Q_n^{-2}(Q_n d_n-1)Q_n^{n-2m}
+Q_n^{-2}d_n^{-1}(1-Q_n)(Q_n d_n)^{n-2m})
\end{align*}
where $\Theta_n=(1-Q_n)(1-Q_n d_n)$. Note that the two ratios~$Q_n$
and~$Q_n d_n$ are different from~1 by our assumptions. In fact,
$Q_n=1$ is equivalent to $\cosh(\sigma\sqrt{T/n})=e^{-rT/n}$ which is
never the case with~$\sigma$ strictly positive; and $Q_n d_n=1$ is
equivalent to $e^{-rT/n}=1$ which is the case only if $r=0$.
Substituting these relations into $\varphi_1$ and $\varphi_2$,
\begin{align*}
\varphi_1 &= \frac{Q_n}{\Theta_n}(1-d_n)\sum_{m=0}^{\lfloor n/2\rfloor}
\binom{n}{m} q_n^m (1-q_n)^{n-m}\\
&\quad +\frac{Q_n}{\Theta_n}(Q_n d_n-1)
\sum_{m=0}^{\lfloor n/2\rfloor}
\binom{n}{m} q_n^m (1-q_n)^{n-m}Q_n^{n-2m}\\
&\quad +\frac{Q_n d_n}{\Theta_n}(1-Q_n)
\sum_{m=0}^{\lfloor n/2\rfloor}
\binom{n}{m} q_n^m (1-q_n)^{n-m}(Q_n d_n)^{n-2m},
\end{align*}
and
\begin{align*}
\varphi_2 &= \frac{Q_n^2}{\Theta_n}(1-d_n)\sum_{m=0}^{\lfloor n/2\rfloor-1}
\binom{n}{m} q_n^m (1-q_n)^{n-m}\\
&\quad +\frac{1}{\Theta_n}(Q_n d_n-1)
\sum_{m=0}^{\lfloor n/2\rfloor-1}
\binom{n}{m} q_n^m (1-q_n)^{n-m}Q_n^{n-2m}\\
&\quad +\frac{1}{\Theta_n d_n}(1-Q_n)
\sum_{m=0}^{\lfloor n/2\rfloor-1}
\binom{n}{m} q_n^m (1-q_n)^{n-m}(Q_n d_n)^{n-2m}.
\end{align*}
Now, by definition of~$Q_n$,
\begin{equation*}
\sum\binom{n}{m} q_n^m (1-q_n)^{n-m}Q_n^{n-2m}
=\sum\binom{n}{m} (1-q_n)^m q_n^{n-m};
\end{equation*}
moreover, since $(1-q_n)Q_n d_n=p_n e^{-rT/n}$ and
$q_n/(Q_n d_n)=(1-p_n)e^{-rT/n}$ we get that
\begin{equation*}
\begin{split}
\sum\!\binom{n}{m} q_n^m (1-q_n)^{n-m}(Q_n d_n)^{n-2m}
&\!=\!\sum\!\binom{n}{m}\!\Big(\frac{q_n}{Q_n d_n}\Big)^m
\!((1-q_n)Q_n d_n)^{n-m}\\
&\!=\!e^{-rT}\sum\!\binom{n}{m} (1-p_n)^m p_n^{n-m}.
\end{split}
\end{equation*}
The value~(\ref{def-V}) therefore becomes
\begin{equation}\label{simp-V}
V_n(0,0)=\frac{Q_n}{\Theta_n}(1-d_n)\phi_1-\frac{1}{1-Q_n}\phi_2
+\frac{e^{-rT}}{1-Q_n d_n}\phi_3,
\end{equation}
where 
\begin{itemize}
\setlength{\itemsep}{1mm}
\item[$\phi_1=$\!] $\mathcal{B}_{n,q_n}(\lfloor n/2\rfloor)
-Q_n\,\mathcal{B}_{n,q_n}(\lfloor n/2\rfloor-1)$,
\item[$\phi_2=$\!] $Q_n\,\mathcal{B}_{n,1-q_n}(\lfloor n/2\rfloor)
-\mathcal{B}_{n,1-q_n}(\lfloor n/2\rfloor-1)$,
\item[$\phi_3=$\!] $Q_n d_n\,\mathcal{B}_{n,1-p_n}(\lfloor n/2\rfloor)
-u_n\,\mathcal{B}_{n,1-p_n}(\lfloor n/2\rfloor-1)$,
\end{itemize}
with $\mathcal{B}_{n,p}$ the cumulative distribution function of the
binomial distribution with parameters $n$ and $p$. In the following we
will need the asymptotics of $p_n$, $q_n$ and~$Q_n$ given by
\[
p_n=\frac{e^{rT/n}-e^{-\sigma\sqrt{T/n}}}
{e^{\sigma\sqrt{T/n}}-e^{-\sigma\sqrt{T/n}}},
\quad
q_n=p_n e^{\sigma\sqrt{T/n}} e^{-rT/n},
\]
and~$Q_n=q_n/(1-q_n)$. They are
\begin{equation}\label{da-pn}
p_n=\frac{1}{2}
	+\frac{\alpha_2}{\sqrt{n}}
	+\frac{\gamma_2}{n^{3/2}}
	+O\Big(\frac{1}{n^{5/2}}\Big),
\end{equation}
where
\begin{equation*}
\alpha_2=\frac{2r-\sigma^2}{4\sigma}\sqrt{T},\quad 
\gamma_2=\frac{12r^2-4r\sigma^2+\sigma^4}{48\sigma}\,T^{3/2},
\end{equation*}
\begin{equation}\label{da-qn}
q_n=\frac{1}{2}
	+\frac{\alpha_1}{\sqrt{n}}
	+\frac{\gamma_1}{n^{3/2}}
	+O\Big(\frac{1}{n^{5/2}}\Big),
\end{equation}
where
\begin{equation*}
\alpha_1=\frac{2r+\sigma^2}{4\sigma}\sqrt{T},\quad 
\gamma_1=-\frac{12r^2+4r\sigma^2+\sigma^4}{48\sigma}\,T^{3/2}.
\end{equation*}
Here we have chosen the indices of~$\alpha_1$ and~$\alpha_2$ in order
to be consistent with the literature. We deduce from this that
\begin{equation}\label{da-theta1}
Q_n = 1+\frac{4\alpha_1}{\sqrt{n}}
	+\frac{8\alpha^2_1}{n}\\
	+\frac{16\alpha_1^3+4\gamma_1}{n^{3/2}}
	+O\Big(\frac{1}{n^2}\Big).
\end{equation}
We now expand each of the terms of~$V_n(0,0)$ in~(\ref{simp-V}). We
begin with the three coefficients.  We obtain by~(\ref{da-theta1}) and
a Taylor expansion for~$d_n$ about~$0$ that
\begin{equation}\label{terme2}
\frac{1}{1-Q_n}=-\frac{1}{4\alpha_1}\sqrt{n}+\frac{1}{2}
+\frac{\gamma_1}{4\alpha^2_1}\frac{1}{\sqrt{n}}+O\Big(\frac{1}{n}\Big),
\end{equation}
\begin{equation}\label{terme3}
\frac{1}{1-Q_n d_n}=-\frac{\sigma}{2r\sqrt{T}}\sqrt{n}+\frac{1}{2}
+\frac{\sigma^3\sqrt{T}}{24r}\frac{1}{\sqrt{n}}+O\Big(\frac{1}{n}\Big).
\end{equation}
From~(\ref{terme2}) and~(\ref{terme3}) we obtain the third coefficient
as
\begin{equation}\label{terme1}
\begin{split}
\frac{Q_n}{\Theta_n}(1-d_n)
&=\frac{1}{1-Q_n}-\frac{1}{1-Q_n d_n}\\
&=\frac{\sigma^2}{8r\alpha_1}\sqrt{n}
-\Big(\frac{\alpha_1\sigma^2}{6r}
+\frac{r^2 T^{3/2}}{24\alpha_1^2\sigma}\Big)
\frac{1}{\sqrt{n}}+O\Bigl(\frac{1}{n}\Bigr).
\end{split}
\end{equation}
Lemma~\ref{lemme-nouveau} gives us the asymptotics for the
complementary cumulative distribution function, that is, for
$1-\mathcal{B}_{n,p_n}(j_n-1)$. We use this link to obtain the
expansions for~$\phi_1$, $\phi_2$ and~$\phi_3$.  We apply the lemma
with $j_n=\lfloor n/2\rfloor +1=n/2+1/2+b_n$, where $b_n=1/2$ if
$n$~is even and 0 else, and $j_n=\lfloor n/2\rfloor=n/2+1/2+b_n^*$,
where $b_n^*=b_n-1$. We obtain after a long series of simple
calculations and simplifications
\begin{equation}
\begin{split}
\phi_1&=\Big(2\frac{e^{-2\alpha^2_1}}{\sqrt{2\pi}}
+4\alpha_1(\Phi(2\alpha_1)-1)\Big)\frac{1}{\sqrt{n}}\\
&\quad+\Big(4\alpha_1\frac{e^{-2\alpha^2_1}}{\sqrt{2\pi}}
+8\alpha_1^2(\Phi(2\alpha_1)-1)\Big)\frac{1}{n}\\
&\quad +\Big((2b_n-\frac{3}{2}+6\alpha^2_1)\frac{e^{-2\alpha^2_1}}{\sqrt{2\pi}}
+(16\alpha_1^3+4\gamma_1)(\Phi(2\alpha_1)-1\Big)
\frac{1}{n^{3/2}}\\
&\quad +O\Big(\frac{1}{n^2}\Big),
\end{split}
\end{equation}
\begin{equation}
\begin{split}
  \phi_2&=\Big(2\frac{e^{-2\alpha^2_1}}
  {\sqrt{2\pi}}+4\alpha_1\Phi(2\alpha_1)\Big)\frac{1}{\sqrt{n}}
  +\Big(4\alpha_1\frac{e^{-2\alpha^2_1}}
  {\sqrt{2\pi}}+8\alpha_1^2\Phi(2\alpha_1)\Big)\frac{1}{n}\\
  &\quad
  +\Big((2b_n-\frac{3}{2}+6\alpha^2_1)\frac{e^{-2\alpha^2_1}}{\sqrt{2\pi}}
  +(16\alpha_1^3+4\gamma_1)\Phi(2\alpha_1)\Big)\frac{1}{n^{3/2}}\\
  &\quad +O\Big(\frac{1}{n^2}\Big),
\end{split}
\end{equation}
and
\begin{equation}\label{da-phi3}
\begin{split}
\phi_3&=\Big(2\frac{e^{-2\alpha^2_2}}
{\sqrt{2\pi}}+4\alpha_2\Phi(2\alpha_2)\Big)\frac{1}{\sqrt{n}}
+\Big(4\alpha_1\frac{e^{-2\alpha^2_2}}
{\sqrt{2\pi}}+8\alpha_1\alpha_2\Phi(2\alpha_2)\Big)\frac{1}{n}\\
&\quad +\Big((2b_n-\frac{3}{2}
+4\alpha^2_1+2\alpha^2_2)\frac{e^{-2\alpha^2_2}}{\sqrt{2\pi}}\\
&\qquad +\Big(\frac{2r^3}{\sigma^3}+
\frac{\sigma r}{6}-\frac{\sigma^3}{6}\Big)T^{3/2}
\Phi(2\alpha_2)\Big)\frac{1}{n^{3/2}}
+O\Big(\frac{1}{n^2}\Big).
\end{split}
\end{equation}
The conclusion can be obtained after several simplications with
(\ref{simp-V}), (\ref{terme2})-(\ref{da-phi3}) and using that
$e^{-2\alpha_2^2}e^{-rT}=e^{-2\alpha_1^2}$. This result is
\begin{align*}
V_n(0,0)&=\Bigl(1+\frac{\sigma^2}{2r}\Bigr)\Phi(a_1)
-e^{-rT}\Bigl(1-\frac{\sigma^2}{2r}\Bigr)\Phi(a_2)
-\frac{\sigma^2}{2r}\\
&\quad + \frac{\sigma\sqrt{T}}{2}
\Bigl[\Bigl(1+\frac{\sigma^2}{2r}\Bigr)\Phi(a_1)
-e^{-rT}\Bigl(1-\frac{\sigma^2}{2r}\Bigr)\Phi(a_2)
-\Bigl(1+\frac{\sigma^2}{2r}\Bigr)\Bigr]\frac{1}{\sqrt{n}}\\
&\quad + \Bigg[\frac{\sigma^2 T}{12}
\Bigl[\Bigl(3+\frac{\sigma^2}{2r}\Bigr)\Phi(a_1)
-e^{-rT}\Bigl(3-\frac{\sigma^2}{2r}\Bigr)\Phi(a_2)
-\Bigl(3+\frac{\sigma^2}{2r}\Bigr)\Bigr]\\
&\qquad +\frac{\sigma\sqrt{T}}{2}\frac{e^{-a_1^2/2}}{\sqrt{2\pi}}\Bigg]\frac{1}{n}
+O\Bigl(\frac{1}{n^{3/2}}\Bigr)
\end{align*}
with $a_1=2\alpha_1$ and $a_2=2\alpha_2$. So the value of the call can
be deduced with $C^{fl}_n=S_0 V_n(0,0)$, and the theorem follows
with~($\ref{formule-BS-call})$.
\end{proof}

We can give a similar result for the European lookback put option with
floating strike.
\begin{theoreme}\label{thm-put}
  In the n-period CRR binomial model, let~$S_0$ be the initial stock
  price, $r$~the risk free interest rate, $\sigma$~the volatility and
  $T$~the time to maturity. If $r\neq 0$ then the asymptotic formula
  for the price of the European lookback put option with floating
  strike is
\[
\begin{split}
P_{n}^{fl}&=P^{fl}_{BS}-\frac{\sigma\sqrt{T}}{2}
(P^{fl}_{BS}+S_0)\frac{1}{\sqrt{n}}\\
&\quad +\Big[\frac{\sigma^2 T}{12}(P_{BS}^{fl}
\!+\!2S_0 \Big[\Phi(a_1)\!-\!e^{-rt}(\Phi(a_2)\!-\!1)\!+\!\frac{1}{2}\Big])
\!+\!S_0\frac{\sigma\sqrt{T}}{2}\frac{e^{-a_1^2/2}}{\sqrt{2\pi}}\Big]
\frac{1}{n}\\
&\quad +O\Bigl(\frac{1}{n^{3/2}}\Bigr),
\end{split}
\]
where $P^{fl}_{BS}$ is the price given by the Black-Scholes model,
$\Phi$ the standard normal cumulative distribution function,
$a_1=(r/\sigma+\sigma/2)\sqrt{T}$ and
$a_2=(r/\sigma-\sigma/2)\sqrt{T}$.
\end{theoreme}

\begin{proof}[Sketch of proof]
  As the approach is quite similar to the previous result, we give
  only the important intermediate steps.  By
  Theorem~\ref{CV-formulas}, the put price~$P^{fl}_n$ is~$S_0
  W_n(0,0)$ with~$W_n(0,0)$ defined in~(\ref{def-W}).  Rewriting the
  double sum as before, we obtain that
\begin{equation}
W_n(0,0)=\frac{1}{\Theta_n^*}(u_n-1)\phi_4-\frac{1}{1-Q_n}\phi_5
+\frac{e^{-rT}}{u_n-Q_n}\phi_6,
\end{equation}
where $\Theta_n^*=(Q_n-u_n)(Q_n-1)$,
$\phi_4=Q_n\,\mathcal{B}_{n,1-q_n}(\lfloor n/2\rfloor)
-\mathcal{B}_{n,1-q_n}(\lfloor n/2\rfloor-1)$,
$\phi_5=\mathcal{B}_{n,q_n}(\lfloor n/2\rfloor)
-Q_n\,\mathcal{B}_{n,q_n}(\lfloor n/2\rfloor-1)$ and
$\phi_6=u_n\,\mathcal{B}_{n,p_n}(\lfloor n/2\rfloor) -Q_n
d_n\,\mathcal{B}_{n,p_n}(\lfloor n/2\rfloor-1)$, with
$\mathcal{B}_{n,p}$ the cumulative distribution function of the
binomial distribution with parameters $n$ and $p$. To write this
expression as desired it remains to apply Lemma~$\ref{lemme-nouveau}$
six times and simplify.
\end{proof}

\begin{rmq}\label{remark_r0}
  In both results, we have obtained the asymptotic behavior for a non
  zero risk free interest rate. But we can show from
  Lemma~\ref{lemme-nouveau} that the asymptotic expansion in powers
  of~$n^{-1/2}$ to the order~$1$ is exactly the same for the
  case~$r=0$, that is to say
\begin{equation}\label{form_call_r0}
C_{n}^{fl}=C^{fl}_{BS}+\frac{\sigma\sqrt{T}}{2}
(C^{fl}_{BS}-S_0)\frac{1}{\sqrt{n}}+O\Big(\frac{1}{n}\Big),
\end{equation}
where $C^{fl}_{BS}$ is the price given by the Black-Scholes model
for~$r=0$, that is,
\[
C^{fl}_{BS}=S_0\frac{\sigma\sqrt{T}}{\sqrt{2\pi}}e^{-\sigma^2 T/8}
+S_0\Phi\Big(\frac{\sigma\sqrt{T}}{2}\Big)
-S_0\Phi\Big(\frac{-\sigma\sqrt{T}}{2}\Big)
\Big(1+\frac{\sigma^2 T}{2}\Big).
\]
The latter was obtained by Babbs~\cite{babbs} from the price for
$r\neq 0$ by passing to the limit. While, for $r=0$, $Q_n$~remains
different from~$1$, we get that $Q_n d_n=1$ so that the value of the
second geometric series in~(\ref{decomposition_sommes}) is $n-2m+1$.
Thus the writing of~$V_n(0,0)$ will be different and its coefficients
are of a higher order than those obtained in~(\ref{simp-V}). As a
result, similar calculations as before lead to~(\ref{form_call_r0}),
but we do not get the term in~$n^{-1}$. In order to calculate that
term, we would have to improve Lemma~\ref{lemme-nouveau}: we would
need an asymptotic expansion of the sum in powers of~$n^{-1/2}$ to the
order~4.

In the same way the value of the put in the case $r=0$ can be
approximated by
\[
P_{n}^{fl}=P^{fl}_{BS}-\frac{\sigma\sqrt{T}}{2}
(P^{fl}_{BS}+S_0)\frac{1}{\sqrt{n}}+O\Big(\frac{1}{n}\Big),
\]
where $P^{fl}_{BS}$ is the price given by the Black-Scholes model
for~$r=0$, that is,
\[
P^{fl}_{BS}=C^{fl}_{BS}+\frac{S_0\sigma^2 T}{2};
\]
the latter can be deduced as for the call.
\end{rmq}
\section{Convergence of the delta}\label{section-delta}
An important way to control the risks related to exotic options is
given by the Greek values that give information on price variation as
one varies the parameters. One of the basic sensitivities is the first
order Greek delta. This is the first partial derivative of the option
value with respect to the value of the underlying asset.  The delta
has been studied for lookback options by Pedersen~\cite{pedersen},
Bernis et al.~\cite{bernis} and Gobet and
Kohatsu-Higa~\cite{gobet_kohatsu} but without using asymptotic
expansions; this we will do in this section. From the definition of
the delta, we quickly get its value for the call from the
Goldman-Sosin-Gatto~\cite{goldman_sosin} formula:

\begin{equation}
\Delta_{BS}^{fl}=\Big(1+\frac{\sigma^2}{2r}\Big)\Phi(a_1)
	-e^{-rT}\Big(1-\frac{\sigma^2}{2r}\Big)\Phi(a_2)
	-\frac{\sigma^2}{2r}.
\end{equation}
In practice, in order to provide an estimate of the delta from the
Cox-Ross-Rubinstein model, we use
\begin{equation}\label{def-delta-gen}
\Delta_n^{fl}=\frac{C^{fl}_n(1,1)-C^{fl}_n(0,1)}{S_0 u_n-S_0 d_n},
\end{equation}
where $C^{fl}_n(1,1)$ is the price of the call associated with the
node~(1,1) in the $n$th tree and $C^{fl}_n(0,1)$ the price of the call
associated with the node~(0,1) in the same tree, as was described in
Hull~\cite{hull_white}. Applying the Cheuk-Vorst lattice, we can
deduce that this is equivalent to
\begin{equation}\label{def-delta}
\Delta_n^{fl}=\frac{u_n V_n(1,1)-d_n V_n(0,1)}{u_n-d_n}.
\end{equation}
This approximation for delta is calculated at time~$T/n$ but it is
used as an estimate for time zero.

\begin{theoreme}\label{thm-delta}
  In the n-period CRR binomial model, let~$S_0$ be the initial stock
  price, $r$~the risk free interest rate, $\sigma$~the volatility and
  $T$~the time to maturity. If $r\neq 0$ then the asymptotic formula
  for the delta of the European lookback call option with floating
  strike is
\[
\begin{split}
\Delta_{n}^{fl}&=\Delta^{fl}_{BS}
	-\Big[\frac{\sigma^2}{2r}a_1\Phi(-a_1)
	-e^{-rT}\Big(1-\frac{\sigma^2}{2r}\Big)a_2\Phi(a_2)
	-\frac{e^{-a_1^2/2}}{\sqrt{2\pi}}\Big]\frac{1}{\sqrt{n}}\\
&\quad +O\Bigl(\frac{1}{n}\Bigr),
\end{split}
\]
where $\Delta^{fl}_{BS}$ is the delta obtained by derivation from the
Black-Scholes formula price, $\Phi$ the standard normal cumulative
distribution function, $a_1=(r/\sigma+\sigma/2)\sqrt{T}$ and
$a_2=(r/\sigma-\sigma/2)\sqrt{T}$.
\end{theoreme}

\begin{proof}
  To determine this approximation from~(\ref{def-delta}), we first
  expand~$V_n(0,1)$. We deduce this in the same way as~$V_n(0,0)$.
  The number of paths to arrive at a level~$j$ from the node~$(0,1)$
  in a tree with $n$~steps is the same as that to arrive at this level
  from the node~$(0,0)$ in a tree with $n-1$~steps. From this, we
  easily deduce that
\begin{equation}\label{formulegen-Vn01}
V_n(0,1)=\sum_{j=0}^{n-1}(1-u_n^{-j})
	\sum_{k=j}^{l_3} \Lambda_{j,k,n-1}\,q_n^k (1-q_n)^{n-1-k},
\end{equation}
where~$l_3=\lfloor(n-1+j)/2\rfloor$, see Theorem~\ref{CV-formulas}.
To obtain the asymptotic expansion of~$V_n(0,1)$, we interchange the
sums as it was done previously. After this, using the formula for the
geometric series and simplifying the different sums, we have
\begin{equation}\label{simp-V01}
V_n(0,1)=\frac{Q_n}{\Theta_n}(1-d_n)\phi_7-\frac{1}{1-Q_n}\phi_8
+\frac{e^{-rT}}{1-Q_n d_n}\phi_9,
\end{equation}
where 
\begin{itemize}
\setlength{\itemsep}{1mm}
\item[$\phi_7=$\!] $\mathcal{B}_{n-1,q_n}(\lfloor (n-1)/2\rfloor)
-Q_n\,\mathcal{B}_{n-1,q_n}(\lfloor (n-1)/2\rfloor-1)$,
\item[$\phi_8=$\!] $Q_n\,\mathcal{B}_{n-1,1-q_n}(\lfloor (n-1)/2\rfloor)
-\mathcal{B}_{n-1,1-q_n}(\lfloor (n-1)/2\rfloor-1)$,
\item[$\phi_9=$\!] $Q_n d_n\,\mathcal{B}_{n-1,1-p_n}(\lfloor (n-1)/2\rfloor)
-u_n\,\mathcal{B}_{n-1,1-p_n}(\lfloor (n-1)/2\rfloor-1)$,
\end{itemize}
with $\mathcal{B}_{n,p}$ the cumulative distribution function of the
binomial distribution with parameters~$n$ and~$p$.

Next we want to expand~$V_n(1,1)$
that is, by definition,
\begin{equation}\label{simp-V11-def}
V_n(1,1)=\sum_{j=0}^{n}(1-u_n^{-j})P(V^*_{1,n}=1, V^*_{n,n}=j),
\end{equation}
see Theorem~\ref{CV-formulas}.  It remains to evaluate the
probabilities for these values of~$j$. The paths leaving from~$(1,1)$
can be divided into three parts. The first part contains all paths
that arrive at level~$n$ or~$n-2$. For these paths we are sure that
there had been no change for the minimal value.  Note that it is
impossible to reach level~$n-1$ leaving from~$(1,1)$. If the path
reaches the level~$n$, it is after exactly $n-1$~upward jumps and no
downward jump. If the path reaches the level~$n-2$, it is after
exactly $n-2$~upward jumps and one downward jump. The number of such
paths is obtained with the binomial coefficient. The link between the
reached level~$j$ and the number of upward jumps~$k$ in this binomial
tree is $j=-n+2+2k$.  Considering the number of upward jumps rather
than the level reached, we have that
\begin{equation}\label{simp-V11-decomp1}
\begin{split}
&\sum_{j=n-2}^{n}(1-u_n^{-j})P(V^*_{1,n}=1, V^*_{n,n}=j)\\
&\qquad=\sum_{k=n-2}^{n-1}(1-u_n^{n-2-2k})\binom{n-1}{k}q_n^k (1-q_n)^{n-1-k}.
\end{split}
\end{equation}

The second part contains all other paths reaching a final node without
changing the minimal value. Such paths reach final nodes at even level
if $n$~is even and at odd level if $n$~is odd apart from the
levels~$n$ and~$n-2$ considered in the first case. These levels are
achieved with a number of upward jumps between $\lfloor(n-1)/2\rfloor$
for the lowest level and~$n-3$ for the highest.  For a fixed level,
the number of paths is obtained from the binomial coefficient from
which we subtract the number of paths changing the minimum value. This
last number is obtained by using the reflection principle (identical
to the reasoning for barrier options, see
Lyuu~\cite[pp.~234-242]{lyuu}). We apply it as if we were in the
situation where we begin two levels above the barrier (because we can
\textit{touch} the minimal value but we cannot go down just after) and
we finish $j+1$ levels above the barrier after $n-1$ steps. Then the
number of \textit{lost} paths turns out to be
\begin{equation}
\binom{n-1}{\frac{n-4-j}{2}}=\binom{n-1}{n-3-k}=\binom{n-1}{k+2}.
\end{equation}
We deduce that
\begin{equation}\label{simp-V11-decomp2}
\begin{split}
&\sum_{j=0}^{n-3}(1-u_n^{-j})P(V^*_{1,n}=1, V^*_{n,n}=j,\min_{m\leq n}S_m=S_0)\\
&\qquad=\sum_{k=l_4}^{n-3}(1-u_n^{n-2-2k})
	\Big[\binom{n-1}{k}-\binom{n-1}{k+2}\Big]
	q_n^k (1-q_n)^{n-1-k},
\end{split}
\end{equation}
with~$l_4=\lfloor(n-1)/2\rfloor$.

The third part contains all paths reaching a final node after changing
the minimal value. All of these paths arrive in the partial
Cheuk-Vorst lattice whose initial node is at position~(0,3). We need
the following identity:
\begin{equation}\label{proba-arbre-change}
\begin{split}
&P(V^*_{1,n}=1, V^*_{m,n}=j,\min_{m^*\leq m}S_{m^*}<S_0)\\
&\qquad= P\Big(V^*_{0,n}=0, V^*_{m-1,n}=j,
	k\leq \Big\lfloor\frac{m-3+j}{2}\Big\rfloor\Big)
\end{split}
\end{equation}
which can be shown by a detailed analysis of the possible paths as was
done in Appendix~\ref{annexe-CVdensity} for the original Cheuk-Vorst
lattice. Thus we have
\begin{equation}\label{simp-V11-decomp3}
\begin{split}
&\sum_{j=0}^{n-3}(1-u_n^{-j})P(V^*_{1,n}=1, V^*_{n,n}=j,\min_{m\leq n}S_m<S_0)\\
&\qquad=\sum_{j=0}^{n-3}(1-u_n^{-j})
	\sum_{k=j}^{l_3-1} \Lambda_{j,k,n-1}\,q_n^k (1-q_n)^{n-1-k}.
\end{split}
\end{equation}
From~(\ref{simp-V11-decomp1}), (\ref{simp-V11-decomp2})
and~(\ref{simp-V11-decomp3}), we have that~(\ref{simp-V11-def}) can be
rewritten as
\begin{equation}\label{simp-V11-interm}
\begin{split}
V_n(1,1)
&=\sum_{k=l_4}^{n-1}(1-u_n^{n-2-2k})\binom{n-1}{k}q_n^k (1-q_n)^{n-1-k}\\
&\qquad -\sum_{k=l_4}^{n-3}(1-u_n^{n-2-2k})\binom{n-1}{k+2}q_n^k (1-q_n)^{n-1-k}\\
&\qquad +\sum_{j=0}^{n-3}(1-u_n^{-j})
	\sum_{k=j}^{l_3-1} \Lambda_{j,k,n-1}\,q_n^k (1-q_n)^{n-1-k}.
\end{split}
\end{equation}
To obtain the asymptotic expansion of~$V_n(1,1)$, we apply the same
ideas as it was done previously and we obtain
\begin{equation}\label{simp-V11}
\begin{split}
V_n(1,1)=&\phi_{10}-\frac{p_n}{q_n}e^{-rT}\phi_{11}\\
&\quad +\frac{Q_n}{\Theta_n}(1-d_n)\phi_{12}-\frac{Q_n^{-2}}{1-Q_n}\phi_{13}
+\frac{(Q_n d_n)^{-2}e^{-rT}}{1-Q_n d_n}\phi_{14},
\end{split}
\end{equation}
where 
\begin{itemize}
\setlength{\itemsep}{1mm}
\item[$\phi_{10}=$\!] $\mathcal{B}^*_{n-1,q_n}(n-2-\lfloor n/2\rfloor)
-Q_n^{-2}\,\mathcal{B}^*_{n-1,q_n}(n-\lfloor n/2\rfloor)$,
\item[$\phi_{11}=$\!] $\mathcal{B}^*_{n-1,p_n}(n-2-\lfloor n/2\rfloor)
-P_n^{-2}\,\mathcal{B}^*_{n-1,p_n}(n-\lfloor n/2\rfloor)$,
\item[$\phi_{12}=$\!] $\mathcal{B}_{n-1,q_n}(\lfloor (n-1)/2\rfloor)
-Q_n\,\mathcal{B}_{n-1,q_n}(\lfloor (n-1)/2\rfloor-1)$,
\item[$\phi_{13}=$\!] $Q_n\,\mathcal{B}_{n-1,1-q_n}(\lfloor (n-1)/2\rfloor)
-\mathcal{B}_{n-1,1-q_n}(\lfloor (n-1)/2\rfloor-1)$,
\item[$\phi_{14}=$\!] $Q_n d_n\,\mathcal{B}_{n-1,1-p_n}
(\lfloor (n-1)/2\rfloor)-u_n\,\mathcal{B}_{n-1,1-p_n}(\lfloor 
(n-1)/2\rfloor-1)$,
\end{itemize}
with $\mathcal{B}_{n,p}$ the cumulative distribution function of the
binomial distribution with parameters~$n$ and~$p$, and
$\mathcal{B}^*_{n,p}$ its complementary, and $P_n=p_n/(1-p_n)$.

We now expand each of the terms of~(\ref{def-delta}).  As for the
coefficients, we have in addition to~(\ref{da-theta1})
to~(\ref{terme1}) that
\begin{equation}\label{da-pnexpm2}
P_n^{-2}=1-\frac{8\alpha_2}{\sqrt{n}}+O\Big(\frac{1}{n}\Big),
\end{equation}
\begin{equation}\label{da-qnexpm2}
Q_n^{-2}=1-\frac{8\alpha_1}{\sqrt{n}}+O\Big(\frac{1}{n}\Big),
\end{equation}
and
\begin{equation}\label{da-pnsurqn}
\frac{p_n}{q_n}=1-\frac{\sigma\sqrt{T}}{\sqrt{n}}
	+O\Big(\frac{1}{n}\Big).
\end{equation}
The fact that
\begin{equation*}
\frac{1}{2}+\frac{\alpha}{\sqrt{n}}
	+\frac{\gamma}{n^{3/2}}+O\Big(\frac{1}{n^{5/2}}\Big)
=\frac{1}{2}+\frac{\alpha}{\sqrt{n-1}}
	+\frac{\gamma-\alpha/2}{(n-1)^{3/2}}+O\Big(\frac{1}{(n-1)^{5/2}}\Big)
\end{equation*}
for any $\alpha,\gamma$ and
\begin{equation*}
\begin{split}
&\Phi(A)+\frac{B}{\sqrt{n-1}}+\frac{C}{n-1}+\frac{D}{(n-1)^{3/2}}
	+O\Big(\frac{1}{(n-1)^2}\Big)\\
&\qquad\qquad =\Phi(A)+\frac{B}{\sqrt{n}}+\frac{C}{n}+\frac{D+B/2}{n^{3/2}}
	+O\Big(\frac{1}{n^2}\Big),
\end{split}
\end{equation*}
for any $A$, $B$, $C$, $D$, allows us to apply
Lemma~\ref{lemme-nouveau} also to $\phi_9$,
$\phi_{10},\dots,\phi_{14}$.  Therefore we obtain
from~(\ref{da-theta1}) to~(\ref{terme1}), (\ref{simp-V01}),
(\ref{simp-V11}) to~(\ref{da-pnsurqn}) that
\begin{equation}\label{da-numgamma}
\begin{split}
&u_n V_n(1,1)-d_n V_n(0,1)=
	2\Delta_{BS}^{fl}\frac{\sigma\sqrt{T}}{\sqrt{n}}\\
&\qquad-\sigma\sqrt{T}\Big[\frac{\sigma^2}{r}a_1\Phi(-a_1)
		-2e^{-rT}\Big(1-\frac{\sigma^2}{2r}\Big)a_2\Phi(a_2)
		-2\frac{e^{-a_1^2/2}}{\sqrt{2\pi}}\Big]\frac{1}{n}\\
&\qquad+O\Big(\frac{1}{n^{3/2}}\Big).
\end{split}
\end{equation}
It remains to expand the denominator:
\begin{equation}\label{da-denogamma}
\frac{1}{u_n-d_n}=\frac{1}{2\sigma\sqrt{T}}\sqrt{n}
	+O\Big(\frac{1}{n^{1/2}}\Big).
\end{equation}
So, the theorem follows by multiplying~(\ref{da-numgamma})
with~(\ref{da-denogamma}).

\end{proof}
\section{Numerical example}
In this part, we give a numerical illustration of
Theorems~\ref{thm-call}, \ref{thm-put}, Remark~\ref{remark_r0} and
Theorem~\ref{thm-delta}. If we write~$\Pi^{fl}_n$ in the form
of~(\ref{da-options}), we should find that
$(\Pi^{fl}_n-\Pi^{fl}_{BS})\sqrt{n}$ and
$(\Pi^{fl}_n-\Pi^{fl}_{BS}-\Pi_1/\sqrt{n})n$ almost coincide
respectively with $\Pi_1$ and~$\Pi_2$ for large~$n$. This is what we
will compare in the first two tables below. For our illustration, we
take as values $S_0=80$, $\sigma=0.2$, $r=0.08$ and $T=1$, and for the
call the results are:

\vspace*{4mm}
\noindent\begin{tabularx}{\textwidth}{p{0.33\textwidth}*{5}{>{\raggedleft}X}}
\hline
Number of periods~$n$ & 1,000 & 5,000 & 10,000& 50,000&100,000\tabularnewline
\hline
$C^{fl}_n $           &14.7183&14.8304&14.8571&14.8929&14.9014\tabularnewline
$C^{fl}_{BS}$         &14.9219&14.9219&14.9219&14.9219&14.9219\tabularnewline
$(C^{fl}_n-C^{fl}_{BS})\sqrt{n}$                    
                      &-6.4402&-6.4776&-6.4865&-6.4983&-6.5011\tabularnewline
$C_1$                 &-6.5078&-6.5078&-6.5078&-6.5078&-6.5078\tabularnewline
$(C^{fl}_n-C^{fl}_{BS}-C_1/\sqrt{n})n$
                      & 2.1372& 2.1339& 2.1331& 2.1318& 2.1316\tabularnewline
$C_2$                 & 2.1308& 2.1308& 2.1308& 2.1308& 2.1308\tabularnewline
\hline
\end{tabularx}
\vspace*{4mm}

The results obtained for the call with randomly chosen parameters are consistent 
with the first main theorem.
We take the same parameters for the put, and we obtain the following results:

\vspace*{4mm}
\noindent\begin{tabularx}{\textwidth}{p{0.33\textwidth}*{5}{>{\raggedleft}X}}
\hline
Number of periods~$n$ & 1,000 & 5,000 & 10,000& 50,000&100,000\tabularnewline
\hline
$P^{fl}_n $           &10.0271&10.1820&10.2190&10.2687&10.2805\tabularnewline
$P^{fl}_{BS}$         &10.3090&10.3090&10.3090&10.3090&10.3090\tabularnewline
$(P^{fl}_n-P^{fl}_{BS})\sqrt{n}$                    
                      &-8.9151&-8.9790&-8.9942&-9.0145&-9.0193\tabularnewline
$P_1$                 &-9.0309&-9.0309&-9.0309&-9.0309&-9.0309\tabularnewline
$(P^{fl}_n-P^{fl}_{BS}-P_1/\sqrt{n})n$
                      & 3.6627& 3.6695& 3.6710& 3.6729& 3.6734\tabularnewline
$P_2$                 & 3.6744& 3.6744& 3.6744& 3.6744& 3.6744\tabularnewline
\hline
\end{tabularx}
\vspace*{4mm}

The results are also consistent with the result stated 
for the put with randomly chosen parameters.
Next we compare $(C^{fl}_n-C^{fl}_{BS})\sqrt{n}$
and $C_1$ for the call with $r=0$; the other parameters 
remaining the same.

\vspace*{4mm}
\noindent\begin{tabularx}{\textwidth}{p{0.33\textwidth}*{5}{>{\raggedleft}X}}
\hline
Number of periods~$n$ & 1,000 & 5,000 & 10,000& 50,000&100,000\tabularnewline
\hline
$C^{fl}_n $           &11.7748&11.8917&11.9197&11.9571&11.9659\tabularnewline
$C^{fl}_{BS}$         &11.9874&11.9874&11.9874&11.9874&11.9874\tabularnewline
$(C^{fl}_n-C^{fl}_{BS})\sqrt{n}$                    
                      &-6.7233&-6.7664&-6.7767&-6.7903&-6.7935\tabularnewline
$C_1$                 &-6.8013&-6.8013&-6.8013&-6.8013&-6.8013\tabularnewline
\hline
\end{tabularx}

\pagebreak
The results are again consistent with what was announced.
Finally we establish the table for the approximation of the delta with the
same values for~$S_0$, $\sigma$ and~$T$ as previously and we take $r=0.08$ as in the 
two first examples.

\vspace*{4mm}
\noindent\begin{tabularx}{\textwidth}{p{0.33\textwidth}*{5}{>{\raggedleft}X}}
\hline
Number of periods~$n$ & 1,000 & 5,000 & 10,000& 50,000&100,000\tabularnewline
\hline
$\Delta^{fl}_n $      & 0.2003& 0.1927& 0.1909& 0.1885& 0.1879\tabularnewline
$\Delta^{fl}_{BS}$    & 0.1865& 0.1865& 0.1865& 0.1865& 0.1865\tabularnewline
$(\Delta^{fl}_n-\Delta^{fl}_{BS})\sqrt{n}$                    
                      & 0.4360& 0.4392& 0.4400& 0.4410& 0.4413\tabularnewline
$\Delta_1$            & 0.4418& 0.4418& 0.4418& 0.4418& 0.4418\tabularnewline
\hline
\end{tabularx}
\vspace*{4mm}

With this table, we also confirm the expected results.

\section{Conclusion}
In our paper, we give an asymptotic expansion in powers of~$n^{-1/2}$
for the price and the delta of European lookback options with floating
strike evaluated by the binomial tree model. We obtain the convergence
to the Black-Scholes price at the origin, and we procure explicit
formulas for the coefficients of~$n^{-1/2}$ and~$n^{-1}$ if $r$~is not
0; for $r=0$ we see that the coefficient of~$n^{-1/2}$ remains the
same, but we do not get the coefficient for~$n^{-1}$. We also obtain
a first order approximation for the delta. These values were
confirmed on random examples.

Several issues can be proposed directly from our work. The first
possibility is to look at what is happening if the risk free interest
rate is zero and therefore to see if also the coefficient of~$n^{-1}$
is unchanged.

Another possibility is not to restrict the evaluation to time~$0$.
The problem seems difficult since a deduction of the formula from
Cheuk-Vorst will be complicated with levels no longer necessarily
integers. Moreover, the tree does not necessarily start at zero as the
minimum between the origin and the present time may be different from
the current value.

It might also be interesting to develop similar theorems for European
lookback options with fixed strike by using the Cheuk-Vorst lattice
tree described in~\cite{cheuk_vorst}. We give in
Appendix~\ref{annexe-fixed} a formula for the price at emission.
However, this formula no longer
involves the cumulative distribution function of binomial distributions, 
so that the techniques of this paper do not allow us to
deduce an asymptotic expansion.
We hope to return to this problem in a future work.

We note that the goal of our research was to study the rate of
convergence of the CRR binomial model. It would be interesting to
develop alternative models with a~higher rate of convergence (as was
done, for example, by Leisen and Reimer~\cite{leisen_reimer} and
Joshi~\cite{joshi} for vanilla options).

\appendix
\section{An equivalent formula}\label{annexe-equivalence}
In their book, Föllmer and Schied~\cite{follmer_schied} introduce a
method for pricing the lookback put. They consider that the price of
the put is the discounted expectation of $S_{\max}-S_T$. They obtain
that the price is~$S_0 W_n^{fs}$ with
\begin{equation}\label{annexe_put}
W_n^{fs}=-1+e^{-rT}\sum_{j=0}^n u_n^j \sum_{k=j}^l \Lambda_{j,k,n} 
p_n^k (1-p_n)^{n-k},
\end{equation}
where $l=\lfloor (n+j)/2\rfloor$,
$p_n=(e^{rT/n}-d_n)/(u_n-d_n)$ and
\begin{equation*}
\Lambda_{j,k,n}=\binom{n}{k-j}-\binom{n}{k-j-1}
=\frac{2(j-k)+n+1}{n+1}\binom{n+1}{k-j}
\end{equation*}
if $k>j$, and $\Lambda_{j,k,n}=1$ if $k=j$.  The explicit
expression~(\ref{annexe_put}) is only given in the first edition of
the book. Moreover, there is a small mistake which we have corrected
here.  One can easily detect the error in the penultimate line of
their proof.

Similarly we can derive that the price of the call is~$S_0 V_n^{fs}$ with
\begin{equation}\label{annexe_call}
V_n^{fs}=1-e^{-rT}\sum_{j=0}^n d_n^j \sum_{k=j}^l \Lambda_{j,k,n} 
p_n^{n-k} (1-p_n)^{k}.
\end{equation}
The proof of the equality between the two equivalent forms
(\ref{def-V}) and~(\ref{annexe_call}) of the evaluation of the call
(respectively of (\ref{def-W}) and~(\ref{annexe_put}) for the put) is
long and technical. We omit the details.

\section{Probability mass functions for the Cheuk-Vorst Lattice}
\label{annexe-CVdensity}
In this section we will prove Theorem~\ref{CV-formulas}.  First, it
follows from $\langle 8\rangle$, $\langle 11\rangle$ and $\langle
13\rangle$ in~\cite{cheuk_vorst} that $V_n(0,0)$ is given
by~(\ref{def-V}). Moreover, by the arguments in~\cite{cheuk_vorst},
$P(V^*_{0,n}=0, V^*_{m,n}=j)$ is the sum of the product of the
probabilities along all paths leading to level~$j$ at time~$m$; here,
an up carries probability~$q_n$ and a down carries
probability~$1-q_n$.  Thus in order to prove~(\ref{densite-V}) it
suffices to show that the number of paths with exactly $k$~ups
arriving at time~$m$ at the level~$j$ is
\[
\Lambda_{j,k,m}=\binom{m}{k-j}-\binom{m}{k-j-1},
\qquad j=0,1,\dots,m
\]
if $j\leq k\leq \lfloor (m+j)/2\rfloor$, and the number of such paths
is~0 for the other values of~$k$.  Note that $\binom{m}{-1}=0$.

The result is obvious for $j=m-1$ and $j=m$. In this cases there is
only one path leading to the level~$j$, and it has exactly $m-1$ and
$m$~ups respectively, see Figure~\ref{fig-cheuk}.

Moreover, the result is trivial for $m=0$ and~$m=1$.  We will prove
the claim by induction on~$m$, where it suffices to consider
$j=0,1,\dots,m-2$.  We assume that the claim is correct for $m-1$,
$m\geq 2$.  We distinguish two cases.
\begin{enumerate}
\item If $j=0$, there are two downward paths leading to it; they come
  from the nodes located at levels~$0$ or~$1$ in the previous step.

\noindent If $k=0$, the only possible path comes from $j=0$ and so
\[
\Lambda_{0,0,m}=\Lambda_{0,0,m-1}=1.
\]
If $1\leq k\leq\lfloor (m-1)/2\rfloor$, 
then we have
\[
\begin{split}
\Lambda_{0,k,m} &=\Lambda_{0,k,m-1}+\Lambda_{1,k,m-1}\\[1mm]
&= \binom{m-1}{k}-\binom{m-1}{k-1}+\binom{m-1}{k-1}-\binom{m-1}{k-2} \\[1mm]
&= \binom{m}{k}-\binom{m}{k-1}.
\end{split}
\]
If $k=m/2$, which can only happen if $m$~is even, the only possible
path comes from $j=1$ and so
\[
\quad\qquad
\Lambda_{0,\frac{m}{2},m}=\Lambda_{1,\frac{m}{2},m-1}
=\binom{m-1}{\frac{m}{2}-1}-\binom{m-1}{\frac{m}{2}-2}
=\binom{m}{\frac{m}{2}}-\binom{m}{\frac{m}{2}-1}.
\]
Finally, if $k>\lfloor m/2\rfloor$ there are no paths with $k$~ups. 
\item If $1\leq j\leq m-2$, there are two paths leading to it, one
  upward path coming from the node located at level~$j-1$ and one
  downward path coming from the node located at level~$j+1$ at the
  previous step.

\noindent If $0\leq k<j$ there are no paths with $k$~ups.

\noindent If $k=j$, the only possible path comes from $j-1$ and so
\[
\Lambda_{j,j,m}=\Lambda_{j-1,j-1,m-1}=1.
\]
If $j+1\leq k\leq\lfloor (m+j)/2\rfloor$, 
then we have
\[
\begin{split}
\quad\qquad\Lambda_{j,k,m}&=\Lambda_{j-1,k-1,m-1}+\Lambda_{j+1,k,m-1}\\
&=\binom{m-1}{k-j}-\binom{m-1}{k-j-1}+\binom{m-1}{k-j-1}-\binom{m-1}{k-j-2}\\
&=\binom{m}{k-j}-\binom{m}{k-j-1}.
\end{split}
\]
Finally, if $k>\lfloor (m+j)/2\rfloor$ there are no paths with $k$~ups. 
\end{enumerate}
With these two items all cases are covered and the result is proved.

The probability $P(W^*_{0,n}=0, W^*_{m,n}=j)$ can be derived by a
symmetrical reasoning.

\section{Proof of Lemma~\ref{lemme-nouveau}}\label{annexe-proof-lemma-new}

To demonstrate Lemma~\ref{lemme-nouveau}, we will write each term
in~(\ref{formule-lemma}) in its asymptotic expansion in powers
of~$n^{1/2}$ to the order 3. First, we see that
\begin{equation}\label{eg-terme1}
\int_{\xi}^{\infty}e^{-u^2/2}\,\mathrm{d}u\\
=\sqrt{2\pi}\,\Phi(A)+\int_{A}^{-\xi}e^{-u^2/2}\,\mathrm{d}u.
\end{equation}
This writing is motivated by the convergence of~$\xi$ to~$-A$ as
$n\to\infty$. Indeed we have that
\begin{equation}\label{da-xinum}
j_n-n p_n-\frac{1}{2}=(a-\alpha)\sqrt{n}+b_n-\beta+\frac{c-\gamma}{\sqrt{n}}
+\frac{d-\delta}{n}+O\Bigl(\frac{1}{n^{3/2}}\Bigr),
\end{equation}
and
\begin{equation}
p_n(1-p_n)=\frac{1}{4}-\frac{\alpha^2}{n}
-\frac{2\alpha\beta}{n^{3/2}}+O\Big(\frac{1}{n^2}\Big).
\end{equation}
With Taylor expansion about~$1/4$ for the square root of the previous
expression,
\begin{equation}\label{da-xiden}
\sqrt{p_n(1-p_n)}=\frac{1}{2}-\frac{\alpha^2}{n}
-\frac{2\alpha\beta}{n^{3/2}}+O\Big(\frac{1}{n^2}\Big).
\end{equation}
As a consequence of~(\ref{da-xinum}) and~(\ref{da-xiden}) 
the expansion of $\xi$ is
\begin{equation}\label{da-xi1}
\xi=\frac{j_n-np_n-1/2}{\sqrt{np_n(1-p_n)}}
= -A-\frac{B_n}{\sqrt{n}}-\frac{\kappa}{n}-\frac{\lambda}{n^{3/2}}
+O\Bigl(\frac{1}{n^2}\Bigr),
\end{equation}
where $\kappa=2(\alpha^2 A+\gamma-c)$ and 
$\lambda=2(\alpha^2 B_n+2\alpha\beta A+\delta-d)$.
By Taylor expansion about~$A$, 
there is some $\eta$ between~$A$ and~$-\xi$
such that
\begin{equation}
\begin{split}
\int_{A}^{-\xi}\!\!e^{-u^2/2}\,\mathrm{d}u
&=e^{-A^2/2}\Big[(-\xi-A)-\frac{A}{2}(-\xi-A)^2\\
&\quad +\frac{A^2-1}{6}(-\xi-A)^3\Big]
-\frac{\eta (\eta^2-3)}{24}e^{-\eta^2/2}(-\xi-A)^4.
\end{split}
\end{equation}
Because $-\xi-A=O(n^{-1/2})$ we have $(-\xi-A)^4=O(n^{-2})$.
Since the coefficient of $(-\xi-A)^4$ is bounded on~$\mathbb{R}$,
using~(\ref{da-xi1}), the integral becomes
\begin{equation}\label{da-terme11}
\begin{split}
\int_{A}^{-\xi}e^{-u^2/2}\,\mathrm{d}u
&=e^{-A^2/2}\Big[\frac{B_n}{\sqrt{n}}
+\frac{\kappa-AB_n^2/2}{n}\\
&\quad +\frac{\lambda-\kappa AB_n+(A^2-1)B_n^3/6}{n^{3/2}}\Big]
+O\Bigl(\frac{1}{n^2}\Bigr).
\end{split}
\end{equation}
Using~(\ref{eg-terme1}) and~(\ref{da-terme11}), 
the first term~$T_1$ in~(\ref{formule-lemma}) can be written
\begin{equation}\label{da-terme1}
\begin{split}
T_1&=\Phi(A)+\frac{e^{-A^2/2}}{\sqrt{2\pi}}\Big[\frac{B_n}{\sqrt{n}}
+\frac{\kappa-AB_n^2/2}{n}\\
&\quad +\frac{\lambda-\kappa AB_n+(A^2-1)B_n^3/6}{n^{3/2}}\Big]
+O\Bigl(\frac{1}{n^2}\Bigr).
\end{split}
\end{equation}
For the second term, we have that 
$1-2p_n=-2\alpha/\sqrt{n}-2\beta/n+O(n^{-3/2})$ and therefore 
with~(\ref{da-xiden}),
\begin{equation}\label{da-terme21}
\frac{1-2p_n}{\sqrt{p_n(1-p_n)}}=-\frac{4\alpha}{\sqrt{n}}
-\frac{4\beta}{n}
+O\Bigl(\frac{1}{n^{3/2}}\Bigr).
\end{equation}
Using a Taylor expansion about~$-A$ we get
\begin{equation}\label{da-terme22}
(1-\xi^2)e^{-\xi^2/2}
=e^{-A^2/2}\Bigl[1-A^2-\frac{(3-A^2)AB_n}{\sqrt{n}}\Bigr]
+O\Bigl(\frac{1}{n}\Bigr).
\end{equation}
By combining~(\ref{da-terme21}) and~(\ref{da-terme22}), 
we establish that the second term~$T_2$ in~(\ref{formule-lemma}) is
\begin{equation}\label{da-terme2}
T_2=\frac{e^{-A^2/2}}{6\sqrt{2\pi}}
\Big[\frac{4\alpha(1\!-\!A^2)}{n}+
\frac{4(1\!-\!A^2)(\beta\!-\!\alpha AB_n)\!-\!8\alpha AB_n}{n^{3/2}}\Big]
+O\Big(\frac{1}{n^2}\Big).
\end{equation}
From~(\ref{da-xi1}) and~(\ref{da-terme22}), 
we deduce that the last term~$T_3$ in~(\ref{formule-lemma}) is
\begin{equation}\label{da-terme3}
T_3=-\frac{e^{-A^2/2}}{12\sqrt{2\pi}}
\Big[\frac{A(1-A^2)}{n}+
\frac{(1-4A^2+A^4)B_n}{n^{3/2}}\Big]
+O\Big(\frac{1}{n^2}\Big).
\end{equation}
By combining 
Lemma~\ref{lemme-intermediaire}, (\ref{da-terme1}),
(\ref{da-terme2}) and~(\ref{da-terme3}) we have that
\begin{align*}
&\sum_{k=j_n}^{n}\binom{n}{k}p_n^k(1-p_n)^{n-k}
=T_1+T_2+T_3+O\Big(\frac{1}{n^2}\Big)\\
&= \Phi(A)+\frac{e^{-A^2/2}}{\sqrt{2\pi}}\frac{B_n}{\sqrt{n}}
+\frac{e^{-A^2/2}}{\sqrt{2\pi}}
\Big[\kappa
+\frac{2\alpha}{3}(1-A^2)-\frac{A}{12}(1-A^2)-\frac{AB_n^2}{2}\Big]\frac{1}{n}\\
&\quad +\frac{e^{-A^2/2}}{\sqrt{2\pi}}
\Bigl[\lambda-\kappa AB_n-\frac{B_n^3}{6}(1-A^2)
+\frac{2}{3}(1-A^2)(\beta-\alpha AB_n)\\
&\qquad-\frac{4}{3}\alpha AB_n
-\frac{B_n}{12}(1-4A^2+A^4)\Bigr]\frac{1}{n^{3/2}}
+O\Bigl(\frac{1}{n^2}\Bigr).
\end{align*}
The result of Lemma~\ref{lemme-nouveau} is obtained by simplification.

\section{Proof of Lemma~\ref{lemme-intermediaire}}\label{annexe-proof-lemma}

The proof of Lemma~\ref{lemme-intermediaire} is based on the following
result of Lin and Palmer~\cite{lin_palmer}, which in turn was based on
a lemma of Chang and Palmer~\cite{chang_palmer}.
\begin{lemme}
  Provided $p_n=1/2\,+\,O(n^{-1/2})$ and $p_n(1-p_n)=1/4+O(n^{-1})$ as
  $n\to\infty$ and $0\leq j_n\leq n+1$ for $n$ sufficiently large,
  then
\begin{align*}
&\sum_{k=j_n}^{n}\binom{n}{k}p_n^k(1-p_n)^{n-k}
= \frac{1}{\sqrt{2\pi}}\int_{\xi_1}^{\xi_2}e^{-u^2/2}\,\mathrm{d}u\\
&\quad +\frac{1}{6\sqrt{2\pi}}\frac{1-2p_n}{\sqrt{p_n(1-p_n)}}
\Bigl[(1-\xi_2^2)e^{-\xi_2^2/2}-(1-\xi_1^2)e^{-\xi_1^2/2}\Bigr]
\frac{1}{\sqrt{n}}\\
&\quad +\frac{1}{12\sqrt{2\pi}}
\Bigl[\xi_2 e^{-\xi_2^2/2}(\xi_2^2-1)-\xi_1 e^{-\xi_1^2/2}(\xi_1^2-1)\Bigr]
\frac{1}{n}+O\Bigl(\frac{1}{n^2}\Bigr),
\end{align*}
where $\xi_1=\frac{j_n-np_n-1/2}{\sqrt{np_n(1-p_n)}}$
and $\xi_2=\frac{n(1-p_n)+1/2}{\sqrt{np_n(1-p_n)}}$.
\end{lemme}

We will simplify the statement of the result.  Indeed, we see that the
assumption $p_n(1-p_n)=1/4+O(n^{-1})$ is a consequence of the fact
that $|p_n-1/2|\leq C/\sqrt{n}$ because
\begin{equation}\label{ma4}
\Big|p_n(1-p_n)-\frac{1}{4}\Big|=\Big|-\Big(p_n-\frac{1}{2}\Big)^2\Big|
\leq \frac{C^2}{n}.
\end{equation}
Next we have that $n(1-p_n)+1/2=n/2+O(\sqrt{n})$ and, with Taylor
expansion about $1/4$, $\sqrt{p_n(1-p_n)}=1/2+O(n^{-1})$.  As a
consequence we obtain that $\xi_2=\sqrt{n}+O(1)$. For large~$n$,
$\xi_2$ will always be greater than 2 because it tends to
infinity. For all~$u$ greater than 2, we have the relation
$e^{-u^2/2}\leq e^{-u}$. Therefore, for~$n$ sufficiently large,
\begin{equation}\label{da-terme12}
\int_{\xi_2}^{\infty}e^{-u^2/2}\,\mathrm{d}u
\leq \int_{\xi_2}^{\infty}e^{-u}\,\mathrm{d}u=e^{-\xi_2}.
\end{equation}
This gives us that the integral is~$O(n^{-2})$ because $e^{-\xi_2}=O(n^{-j})$ for 
any $j>1$ and in particular for $j=2$. With this argument,
\begin{equation}\label{ma1}
\int_{\xi_1}^{\xi_2}e^{-u^2/2}\,\mathrm{d}u
=\int_{\xi_1}^{\infty}e^{-u^2/2}\,\mathrm{d}u+O\Big(\frac{1}{n^2}\Big).
\end{equation}
In the same way we deduce that
\begin{equation}\label{ma2}
(1-\xi_2^2)e^{-\xi^2_2/2}=O\Big(\frac{1}{n}\Big),
\end{equation}
and
\begin{equation}\label{ma3}
\xi_2(\xi_2^2-1)e^{-\xi^2_2/2}=O\Big(\frac{1}{n}\Big).
\end{equation}
Replacing $\xi_1$ by $\xi$ and using~(\ref{ma4}), (\ref{ma1}),
(\ref{ma2}) and~(\ref{ma3}) and the fact that
$(1-2p_n)/\sqrt{p_n(1-p_n)}=O(n^{-1/2})$ we obtain
Lemma~\ref{lemme-intermediaire}.

\section{Fixed strike case}\label{annexe-fixed}

The payoff function for a European lookback call with fixed strike~$K$
is given by
\[
C^{fx}=\max (\max_{t\leq T}S_t-K,0).
\]
The payoff function for the lookback put with fixed strike~$K$ is
given by
\[
P^{fx}=\max (K-\min_{t\leq T}S_t,0).
\]
Subsequently, we describe the procedure to rewrite the price of the
call; one for the put is entirely similar. We use the same notation,
$X_n(j,m)$, as Cheuk and Vorst for the value associated with a node,
where $j$~refers to the level of the node (unlike these authors, we
consider~$j$ as the absolute value of the level) and $m$~refers to the
time~$mT/n$. At step~$m$, we decompose the option into two parts: a
known payoff at maturity equal to
\begin{equation}
\max(\max_{m^*\leq m}S_{m^*}-K,0),
\end{equation}
and a European lookback call option with fixed strike~$K^*$, where
\begin{equation}
K^*=\max(\max_{m^*\leq m}S_{m^*},K).
\end{equation}
In other words, if we are in a situation of some gain, it is like
having a new option plus money. This means that if we are at level~0
and we have an upward jump, the tree will refer to a coefficient
linked with an option with a higher strike~$K^*$.

If $S_0\geq K$, the payoff at maturity will be at least~$S_0-K$. In
this case, it is equivalent to have an option with strike~$S_0$ and
$S_0-K$ in money at maturity. From this, the tree for this new option
begins at level~0 and we have that
\begin{equation}
C^{fx}_n=S_0 X_n(0,0)+(S_0-K)e^{-rT}.
\end{equation}
The next step is to rewrite~$X_n(0,0)$. All values associated with the
nodes are obtained by backward induction using the expectation with
respect to the probability~$q_n$ of the next two nodes except for the
final nodes and the nodes at level~0. An upward jump at level~0 for
all~$m<n$ provides a final gain of~$S_m(u_n-1)$. We thus obtain
\begin{equation}
\begin{split}
X_n(0,0)
&=\sum_{j=0}^n X_n(j,n)\, P(W^*_{0,n}=0, W^*_{n,n}=j)\\
&\quad+\sum_{m=0}^{n-1} p_n(u_n-1)\, e^{-rT(n-m)/n}\, P(W^*_{0,n}=0, W^*_{m,n}=0),
\end{split}
\end{equation}
where the probability mass function of~$W^*_{m,n}$ is given
by~(\ref{densite-W}). As noted by Cheuk and Vorst, $X_n(j,n)=0$ for
all~$j$. We therefore have that
\begin{equation}
C^{fx}_n=S_0\,p_n(u_n-1)\,e^{-rT}
	\sum_{m=0}^{n-1}e^{rTm/n}
	\sum_{k=0}^{\lfloor m/2\rfloor} \Lambda_{0,k,m} (1-q_n)^k q_n^{m-k}.
\end{equation}

If $S_0<K$, the level where the tree begins is the non-zero positive
number
\begin{equation}
j_0=\frac{\ln(K/S_0)}{\sigma\sqrt{T/n}}.
\end{equation}
We consider in the sequel all values~$n$ such that~$j_0$ is not an
integer because that occurs seldom or never. Because of this, the
different possible levels for the nodes will be integers (after
changing for the first time the moving strike) and real numbers with
the same fractional part as~$j_0$.  In this case, all the values
associated with the nodes are also obtained by backward induction
using the expectation with respect to the probability~$q_n$ of the
next two nodes except for the final nodes, the nodes at level~0 and
for nodes at level~$\{j_0\}\in\mathopen[0,1\mathclose[$. These last
nodes are justified by the fact that an upward jump from this node
will change~$K^*$.  An upward jump at level~$\{j_0\}$ for $m<n$
provides a final gain of~$S_m(u_n-u_n^{\{j_0\}})$. Such a level only
occurs at odd~$m$ if~$\lfloor j_0\rfloor$ is odd and at even~$m$
if~$\lfloor j_0\rfloor$ is even. Of course, the least possible~$m$
is~$\lfloor j_0\rfloor$ because we can increase only one level at a
time. Moreover, the first time that a level~0 exists is~$\lfloor
j_0\rfloor+1$. We thus obtain
\begin{equation}\label{expression-X}
\begin{split}
X_n(j_0,0)
&=\!\!\sum_{m=\lfloor j_0\rfloor+1}^{n-1}\!\!\!\!\! p_n(u_n-1)e^{-rT(n-m)/n}\, 	
	P(W^*_{0,n}\!=\!j_0, W^*_{m,n}\!=\!0)\\
&+\!\!\sum_{m=\lfloor j_0\rfloor}^{n-1}\!\!\! p_n(u_n\!-\!u_n^{\{j_0\}})
	e^{-rT(n-m)/n}\,P(W^*_{0,n}\!=\!j_0, W^*_{m,n}\!=\!\{j_0\}),
\end{split}
\end{equation}
where $P(W^*_{0,n}=j_0, W^*_{m,n}=0)$ corresponds to the probability
to arrive at level~$0$ leaving the node~$(j_0,0)$ after $m$~periods
and $P(W^*_{0,n}=j_0, W^*_{m,n}=\{j_0\})$ to the probability to arrive
at level~$\{j_0\}$ leaving the node~$(j_0,0)$ after $m$~periods
without changing the moving strike. Using a similar idea
as~(\ref{proba-arbre-change}), we have that
\begin{equation}\label{expression-X2}
\begin{split}
P(W^*_{0,n}=j_0, W^*_{m,n}=0)
&=P(W^*_{0,n}=\lfloor j_0\rfloor, W^*_{m,n}=0,\max_{m\leq n}S_m>K) \\
&=\sum_{k=0}^{l_1} \Lambda_{0,k,m} (1-q_n)^k q_n^{m-k},
\end{split}
\end{equation}
where $l_1=\lfloor (m-1-\lfloor j_0\rfloor )/2\rfloor$. Considering
the number of downward jumps rather than the level reached and using
the reflection principle, we have that the second term
of~(\ref{expression-X}) is equal to
\begin{equation}\label{expression-X3}
\sum_{k=0}^{l_2} p_n(u_n-u_n^{\{j_0\}})\, e^{-rT(n-\lfloor j_0\rfloor-2k)/n}\, 
	P(W^*_{0,n}\!=\!j_0, W^*_{m,n}\!=\!\{j_0\}),
\end{equation}
with $l_2=\lfloor (n-1-\lfloor j_0\rfloor )/2\rfloor$ and
\begin{equation}
P(W^*_{0,n}\!=\!j_0, W^*_{m,n}\!=\!\{j_0\})=
	\Lambda_{0,k,n-\lfloor j_0\rfloor-2k}\,(1-q_n)^k q_n^{\lfloor j_0\rfloor +k}.
\end{equation}
We thus obtain an expression for~$C^{fx}_n=S_0\, X_n(j_0,0)$ 
if $S_0<K$ from~(\ref{expression-X}), 
(\ref{expression-X2}) and~(\ref{expression-X3}).

\end{document}